\documentclass{article}
\usepackage{amsfonts}
\usepackage{amsmath}
\usepackage{amssymb}
\usepackage{amsthm}

\newtheorem{theorem}{Theorem}

\newtheorem{corollary}[theorem]{Corollary}

\newtheorem{definition}[theorem]{Definition}

\newtheorem{lemma}[theorem]{Lemma}

\newtheorem{proposition}[theorem]{Proposition}
\newtheorem{remark}[theorem]{Remark}

\title {\textsc{Binary Representation in Multicomplex and Clifford Algebras}}
\author{Derek Courchesne$^1$ $\quad$ and $\quad$ Sébastien Tremblay$^2$}
\date{May 1, 2025}

\begin{document}


\maketitle

\begin{abstract}
Using a binary representation for basis elements of an algebra combined with a framework of multiplier and index functions, a connection has been established between the structure of a large class of algebras and the XOR componentwise operation. This result covers both multicomplex and Clifford algebras of any dimension. Necessary conditions for the existence of a diagonal basis are derived directly from this result, with an algorithm-like method to obtain it through a change-of-coordinates matrix in the commutative case. The same framework is then used to define an algebra of conjugate operations and study its properties.
\end{abstract}


\section{Introduction}

The structure constants of an algebra $\mathcal{A}$ over a field $\mathbb{K}$ are numbers $c_{ijk} \in \mathbb{K}$ used to express how the elements of a basis behave when multiplied together. If $\mathcal{B}:=\{\mathrm{e}_i\} \subset \mathcal{A}$ is such a basis, then the expansion
$$ \mathrm{e}_i*\mathrm{e}_j = \sum_k c_{ijk} \mathrm{e}_k $$
sets the value of the coefficients $c_{ijk}$ for $\mathcal{B}$. These coefficients have been studied for their applications in physics, mostly in the case of Lie algebras \cite{georgi,HumphreysLieAlg}, but variations of this approach can appear in other contexts when trying to express the result of linear operations on basis elements. One notable example is the connections forms $\omega_{ij}$ in differential geometry \cite{oneill} defined, for a fixed tangent vector $v \in T_{\mathbf{p}}(\mathbb{R}^3)$, as the coefficients expressing the covariant derivative on frame elements $\{E_i\}$ in the frame itself:
$$ \nabla_vE_i = \sum_jw_{ij}E_j(\mathbf{p}). $$
Structure constants are at the core of any given algebra $\mathcal{A}$, since isomorphism classes are determined by the structure constants up to a change of basis.

The aim of this article is to start from a similar approach in a simpler case, when the product of basis elements is not an expansion, but directly an element from the same basis (within a factor in $\mathbb{K}$) i.e.
$$ \mathrm{e}_i*\mathrm{e}_j = c_{ij} \mathrm{e}_{r(i,j)} $$
where $r(i,j)$ is an index that depends on both $i$ and $j$. By specializing to this case, we are able to see more deeply how the indices (when enumerating the basis elements $\mathrm{e}_0, \mathrm{e}_1,\ldots$ in a given way) are linked to the product operation. This analysis is performed using both a framework of multiplier and index functions (analog to the structure constants) \cite{multiZ} and a binary representation for the basis elements. From these we were able to obtain general results in a wide variety of algebras, such as Clifford algebras \cite{baylis, Lounesto_2001}, multicomplex numbers \cite{GBPrice, BCHyper, StruppaVajiac2014} and exterior algebras \cite{oneill, bishopGoldberg}. 

The last two sections focus primarily on multiperplex and multicomplex spaces that are the commutative analogs of Clifford algebras \cite{SharmaCommutativeAnalogs} and the only hypercomplex number systems with a diagonal basis. An algorithm-like method to obtain a diagonal basis is derived, and we define the conjugate algebra which is studied in more detail.

\sloppy\section{The multiplier and index functions}
\label{section : multiplier and index vector spaces}

Let $V$ be a finite vector space of dimension $n$ over a field $\mathbb{K}$. We assume that this vector space is equipped with a binary operation $*$ having the following properties:
\begin{enumerate}
    \item $V$ is closed under $*$;
    \item the operation $*$ is $\mathbb{K}$-bilinear on $V$ (making it an algebra over a field);
    \item there is a basis $\mathcal{B}=\{\mathrm e_k \}_{k=0}^{n-1}$ on $V$ such that
\end{enumerate}
$$ 
\mathrm e_p * \mathrm e_q \in \bigcup_{k=0}^{n-1} \mathbb{K} \mathrm e_k, \qquad \forall p,q \in \{ 0,\ldots,n-1 \}, 
$$
where $\mathbb{K} \mathrm e_k$ is defined as the axis $\mathbb{K} \mathrm e_k := \{ x \mathrm e_k \mid x \in \mathbb{K} \}.$
From the bilinearity property, the operation $*$ is fully characterized by its action on the pairs of the basis elements $\mathcal B$ \cite{Kantor}. Let us consider two elements $u,v \in V$ expressed in terms of the basis $\mathcal B$ with the components $x_0, \ldots, x_{n-1}$ and $y_0, \ldots, y_{n-1}$, respectively, then
$$ u * v = \left( \sum_{p=0}^{n-1} x_p \mathrm e_p \right) * \left( \sum_{q=0}^{n-1} y_q \mathrm e_q \right) = \sum_{p,q = 0}^{n-1} x_p y_q (\mathrm e_p * \mathrm e_q). $$
Since the space $V$ is of dimension $n$, we can find a multiplication table for the binary operation $*$, equivalent to an $n \times n$ matrix 
$A_{pq} = \mathrm e_p * \mathrm e_q$. 
The third property allows us to go further in this reasoning. For fixed values of $p$ and $q$, the result of $\mathrm e_p * \mathrm e_q$ is along a given single axis $\mathbb{K} \mathrm e_r$, i.e. $\exists s \in \mathbb{K}$ and $r \in \{ 0,\ldots,n-1 \}$ such that $A_{pq} = \mathrm e_p * \mathrm e_q = s \mathrm e_r. $
The coefficients $s$ and $r$ are called the \textit{multiplier} and the \textit{coefficient index}, respectively (see \cite{multiZ}). For arbitrary values of $p$ and $q$, we can define the functions $s : \{ 0,\ldots,n-1 \}^2 \rightarrow \mathbb{K}$ and $r : \{ 0,\ldots,n-1 \}^2 \rightarrow \{ 0,\ldots,n-1 \}$ such that
\begin{equation} 
    \label{multiplier and index introduction} A_{pq} = \mathrm e_p * \mathrm e_q = s(p,q) \mathrm e_{r(p,q)}. 
\end{equation}
We name them the \textit{multiplier function} $s(p,q)$ and the \textit{index function} $r(p,q)$. Note that since the binary operation $*$ is fully characterized by its multiplication table $A_{pq}$, it is also fully characterized by its multiplier $s(p,q)$ and index $r(p,q)$ functions (or matrices). 

In the following, we will illustrate that the third property is not very restrictive for several important algebras.

\begin{remark} 
    The multiplier and index functions depend on the basis $\mathcal B$, hence different tables for $s(p,q)$ and $r(p,q)$ does not directly imply that two spaces are not isomorphic. However, if those tables are the same and there is already a vector space isomorphism, then both structures have the same binary operation $*$ and we have an algebra isomorphism.
\end{remark}

\subsection{Examples of multiplier and index functions}
\label{subsection : some examples}
In this subsection, we give explicit examples of the multiplier and index functions for many algebras. However, before considering these examples, we introduce two binary functions that will be useful for the rest of the work: the base-2 numeral representation function $\mathrm{bin}_k(n)$ and its inverse $\mathrm{val}_k(b_{k-1},b_{k-2},\ldots,b_0).$

\begin{definition}[Base-2 numeral representation function]
Let $n,k\in \mathbb N$ such that $n< 2^k$. We define the function $\mathrm{bin}_k(n)$ with fixed-length binary representation of 
$n$ as a $k$-tuple of bits (with leading zeros added when needed) by
$$
\begin{array}{rcl}
\mathrm{bin}_k: \mathbb N & \rightarrow & \{0,1\}^k \\*[2ex]
n &\mapsto & \mathrm{bin}_k(n)=(b_{k-1},b_{k-2},\ldots,b_0)
\end{array}
$$
where $n=\sum_{i=0}^{k-1} b_i\,2^i$ and $b_i\in \{0,1\}$.
\end{definition}

\begin{definition}[Inverse of the base-2 numeral representation function]
Let $n,k\in \mathbb N$ such that $n< 2^k$. Then
$$
\begin{array}{rcl}
\mathrm{val}_k: \{0,1\}^k & \rightarrow & \mathbb N \\*
(b_{k-1},b_{k-2},\ldots,b_0) &\mapsto & \mathrm{val}_k(b_{k-1},b_{k-2},\ldots,b_0)=\displaystyle\sum_{i=0}^{k-1} b_i\,2^i.
\end{array}
$$
\end{definition}

Let us now consider the following algebras:
1) the complex numbers, 2) the split-complex numbers, 3) the dual numbers, 4) the quaternions (all covered in \cite{Kantor}), 5) the bicomplex numbers \cite{GBPrice} and finally 6) the Clifford algebra $\mathrm{Cl}_{1,1}$ \cite{Lounesto_2001}. 

\begin{enumerate}
    
\item Let's start with a simple case: the complex number $\mathbb{C}$ seen as a vector space $\mathbb{R}^2$ over the field of real numbers with the basis elements $\mathrm e_0 := (1,0)$ and $\mathrm e_1 := (0,1)$. It is well known \cite{Kantor} that the binary operation, i.e. the commutative product, is defined as
$$ (a,b) \cdot (c,d) := (ac - bd,\ ad + bc), \quad \forall a,b,c,d \in \mathbb{R}. $$
This product is $\mathbb{R}$-bilinear, and we see from the multiplication table
$$ \begin{array}{c | r r}
     \cdot & \mathrm e_0 & \mathrm e_1\\
     \hline
     \mathrm e_0 & \mathrm e_0 & \mathrm e_1 \\
     \mathrm e_1 & \mathrm e_1 & -\mathrm e_0 \\
\end{array} $$
that the product of any pair of basis elements $\{\mathrm e_0,\,\mathrm e_1\}$ is always along a single axis $\mathbb{R}  \mathrm e_0$ or $\mathbb{R} \mathrm e_1$. Thus this multiplication table can be separated into a multiplier part $s(p,q)$ and an index $r(p,q)$ part for $p,q \in \{0,\, 1\}$. We find
$$ \begin{array}{c | c c}
     s & 0 & 1\\
     \hline
     0 & + & + \\
     1 & + & - \\
\end{array} \qquad \text{and} \qquad \begin{array}{c | c c}
     r & 0 & 1\\
     \hline
     0 & 0 & 1 \\
     1 & 1 & 0 \\
\end{array} $$
where we use the convention (in this case and for the other examples that will follow) that $\pm 1$ is represented by $\pm$ in the multiplier table. The functions $s(p,q)$ and $r(p,q)$ can be written in a simple closed-form expression:
$$ s(p,q) = (-1)^{p \wedge q} \qquad \text{and}\qquad r(p,q) = p \oplus q, $$
where $\wedge$ and $\oplus$ are, respectively, the binary logical operators AND and XOR. 

\item For the split-complex numbers $\{a + b\mathrm j \mid a,b \in \mathbb{R},\ \mathrm j^2:=1\}$ with basis elements $\{\mathrm e_0 := 1,\,\mathrm e_1 := \mathrm j \}$, we find
$$ \begin{array}{c | c c}
     \cdot & 1 & \mathrm j\\
     \hline
     1 & 1 & \mathrm j \\
     \mathrm j & \mathrm j & 1 \\
\end{array} \qquad \begin{array}{c | c c}
     s & 0 & 1\\
     \hline
     0 & + & + \\
     1 & + & + \\
\end{array} \qquad \begin{array}{c | c c}
     r & 0 & 1\\
     \hline
     0 & 0 & 1 \\
     1 & 1 & 0 \\
\end{array} $$
where $$ s(p,q) = 1 \qquad \text{and}\qquad r(p,q) = p \oplus q. $$

\item For the dual numbers $\{a + b\epsilon \mid a,b \in \mathbb{R},\ \epsilon^2:=0\}$ with basis elements $\{\mathrm e_0 := 1,\, \mathrm e_1 := \epsilon \}$, we obtain
$$ \begin{array}{c | c c}
    \cdot & 1 & \epsilon\\
     \hline
     1 & 1 & \epsilon \\
     \epsilon & \epsilon & 0 \\
\end{array} \qquad \begin{array}{c | c c}
     s & 0 & 1\\
     \hline
     0 & + & + \\
     1 & + & 0 \\
\end{array} \qquad \begin{array}{c | c c}
     r & 0 & 1\\
     \hline
     0 & 0 & 1 \\
     1 & 1 & 0 \\
\end{array} $$
$$ s(p,q) = 1 - p \wedge q, \quad r(p,q) = p \oplus q. $$

\item For the quaternions, a vector space over $\mathbb{R}$ with basis elements $\{ \mathrm e_0 := 1, \mathrm e_1 := \mathrm i, \mathrm e_2 := \mathrm j, \mathrm e_3 := \mathrm k\}$,
$$ \begin{array}{c | r r r r}
     \cdot & 1 & \mathrm i & \mathrm j & \mathrm k \\
     \hline
     1 & 1 & \mathrm i & \mathrm j & \mathrm k \\
     \mathrm i & \mathrm i & -1 & \mathrm k & -\mathrm j \\
     \mathrm j & \mathrm j & -\mathrm k & -1 & \mathrm i \\
     \mathrm k & \mathrm k & \mathrm j & - \mathrm i & -1 \\
\end{array} \qquad \begin{array}{c | c c c c}
     s & 0 & 1 & 2 & 3 \\
     \hline
     0 & + & + & + & + \\
     1 & + & - & + & - \\
     2 & + & - & - & + \\
     3 & + & + & - & - \\
\end{array} 
\qquad
 \begin{array}{c | c c c c}
     r & 0 & 1 & 2 & 3 \\
     \hline
     0 & 0 & 1 & 2 & 3 \\
     1 & 1 & 0 & 3 & 2 \\
     2 & 2 & 3 & 0 & 1 \\
     3 & 3 & 2 & 1 & 0 \\
\end{array} $$
One can recognize the Klein group in the index table $r(p,q)$. However, let us rewrite this table with the indices $p,q\in \{0,1,2,3\}$ written in their base-$2$ numeral representation using the $\mathrm{bin}_2$ function. Explicitly, we obtain:
$$ \begin{array}{c | c c c c}
     \mathrm{bin}_2(r) & (0,0) & (0,1) & (1,0) & (1,1) \\
     \hline
     (0,0) & (0,0) & (0,1) & (1,0) & (1,1) \\
     (0,1) & (0,1) & (0,0) & (1,1) & (1,0) \\
     (1,0) & (1,0) & (1,1) & (0,0) & (0,1) \\
     (1,1) & (1,1) & (1,0) & (0,1) & (0,0) \\
\end{array} $$
We observe the same pattern for the index function $r(p,q)$ as in the previous example of the dual numbers: a XOR operation $\oplus$ applied componentwise (bitwise) on the base-2 numeral representation of a pair $p,q$ of indices. Indeed, for 
\begin{equation}
\label{binpq}
\mathrm{bin}_2(p)=(p_1,p_2) \qquad \text{and}\qquad  \mathrm{bin}_2(q)=(q_1,q_2),
\end{equation}
where $p_i,q_i\in \{0,1\}$, we find
$$ 
\mathrm{bin}_2(r) = (p_1\oplus q_1,\ p_2\oplus q_2) 
$$
such that
\begin{equation}
\label{rquaternions}
    r(p,q)=\mathrm{val}_2(p_1\oplus q_1,\ p_2\oplus q_2).
\end{equation}


The expression for the multiplier function $s(p,q)$ is given by
$$ 
s(p,q) = \begin{cases}
    1, & \text{for } p=0 \ \text{ or }\ q = 0,\\
    \varepsilon_{pql} - \delta_{pq}, & \text{for } p,q,l \in \{ 1,2,3 \}\ \text{and }\ p \neq l \neq q,
\end{cases} 
$$
where $\varepsilon_{pql}$ represents the the Levi-Civita symbol and $\delta_{pq}$ is the Kronecker delta.

\item For the bicomplex numbers \cite{BCQM1, BCQM2} with basis elements $\{ \mathrm e_0 := 1, \mathrm e_1 := \mathrm i_1, \mathrm e_2 := \mathrm i_2, \mathrm e_3 := \mathrm i_1\mathrm i_2 \}$, where $\mathrm i_1^2 = \mathrm i_2^2 = -1$,
$$ \begin{array}{r | r r r r}
      \cdot & 1 & \mathrm i_1 & \mathrm i_2 & \mathrm i_1 \mathrm i_2 \\
     \hline
     1 & 1 & \mathrm i_1 & \mathrm i_2 & \mathrm i_1\mathrm i_2 \\
     \mathrm i_1 & \mathrm i_1 & -1 & \mathrm i_1\mathrm i_2 & -\mathrm i_2 \\
     \mathrm i_2 & \mathrm i_2 & \mathrm i_1\mathrm i_2 & -1 & -\mathrm i_1 \\
     \mathrm i_1\mathrm i_2 & \mathrm i_1\mathrm i_2 & -\mathrm i_2 & - \mathrm i_1 & 1 \\
\end{array} \qquad \begin{array}{l | l l l l}
     s & 0 & 1 & 2 & 3 \\
     \hline
     0 & + & + & + & + \\
     1 & + & - & + & - \\
     2 & + & + & - & - \\
     3 & + & - & - & + \\
\end{array} \qquad \begin{array}{c | c c c c}
     r & 0 & 1 & 2 & 3 \\
     \hline
     0 & 0 & 1 & 2 & 3 \\
     1 & 1 & 0 & 3 & 2 \\
     2 & 2 & 3 & 0 & 1 \\
     3 & 3 & 2 & 1 & 0 \\
\end{array} $$
The index function $r(p,q)$ is the same as in the quaternions, i.e. given by equation~(\ref{rquaternions}). Considering now the multiplier function $s(p,q)$, again here from equation~(\ref{binpq}) we rewrite $p$ and $q$ in their base-2 numeral representation. Hence, we obtain
$$ 
s(p,q) = (-1)^{\sum_i p_i \wedge q_i} = (-1)^{\sum_i (p \wedge q)_i}. 
$$


\item Finally, we study the Clifford algebra $\mathrm{Cl}_{1,1}$. A basis for this algebra can be written as $\{ \mathrm e_0 :=  1, \mathrm e_1 := \mathrm i, \mathrm e_2 := \mathrm j, \mathrm e_3 := {\mathrm i\mathrm j} \}$ with $\mathrm i^2 = -1$, $\mathrm j^2 = 1$ and ${\mathrm i\mathrm j} = - {\mathrm j\mathrm i}$. We obtain the following tables:
$$ \begin{array}{r | r r r r}
   \cdot & 1 & \mathrm i & \mathrm j & {\mathrm i\mathrm j} \\
     \hline
     1 & 1 & \mathrm i & \mathrm j & {\mathrm i\mathrm j} \\
     \mathrm i & \mathrm i & -1 & {\mathrm i\mathrm j} & -\mathrm j \\
     \mathrm j & \mathrm j & -{\mathrm i\mathrm j} & 1 & -\mathrm i \\
     {\mathrm i\mathrm j} & {\mathrm i\mathrm j} & \mathrm j & \mathrm i & 1 \\
\end{array} \qquad \begin{array}{c | c c c c}
     s & 0 & 1 & 2 & 3 \\
     \hline
     0 & + & + & + & + \\
     1 & + & - & + & - \\
     2 & + & - & + & - \\
     3 & + & + & + & + \\
\end{array} \qquad
\begin{array}{c | c c c c}
     r & 0 & 1 & 2 & 3 \\
     \hline
     0 & 0 & 1 & 2 & 3 \\
     1 & 1 & 0 & 3 & 2 \\
     2 & 2 & 3 & 0 & 1 \\
     3 & 3 & 2 & 1 & 0 \\
\end{array} $$
Again, we get the XOR index function (\ref{rquaternions}) and a multiplier of the form $s(p,q) = \pm 1$. This is not a coincidence, as we will see shortly, but rather a universal property of a specific class of algebras.
\end{enumerate}

\sloppy\section{Algebras with the XOR index function}
\label{section : Commutation property and XOR index function}

\begin{definition}[Hypercomplex algebra]
    \label{XOR algebra : hypercomplex algebra def}
    Let's consider a vector space $\mathbb{K}^{n+1}$ with a basis $\{ 1, u_1, \ldots, u_n \}$ and a binary operation $*$ such that
    \begin{enumerate}
        \item $1 * u_i = u_i = u_i * 1$;
        \item $u_i * u_j = \lambda (u_j * u_i) \quad \text{for}\quad i \neq j$;
        \item $1^2 = 1 \quad \text{and} \quad u_i^2 \in \{ -1, 0 ,1 \}$;
        \item $(u_i*u_j)*u_k = u_i*(u_j*u_k)$ for all $i,j$.
    \end{enumerate}
for all $i,j = 1, \ldots, n$ and $\lambda\in \mathbb K$. The resulting algebra will be called a hypercomplex algebra. 
\end{definition}
This construction encompasses Clifford algebras and a large class of commutative algebras including the hypercomplex number spaces. The elements $u_i, i = 1,\ldots, n$ will be called \textit{principal units} and any multiplication of distinct elements from the basis set will be called a \textit{composite unit}. For the remaining of this section, we also use the juxtaposition instead of explicitly writing the binary operation symbol $*$.
\begin{proposition}
    \label{XOR algebra : commutation property lambda values}
    In a hypercomplex algebra, the commutation property 
    $$ u_iu_j = \lambda u_j u_i \quad \text{for}\quad i \neq j \quad \text{and}\quad  \quad i,j = 1, \ldots, n $$
    implies that either $\lambda = +1$ or $\lambda = -1$.
\end{proposition}
\begin{proof}
    By applying this property twice for $i \neq j$ :
    $$ u_i u_j = \lambda u_j u_i = \lambda^2 u_i u_j \ \Rightarrow \ \lambda^2 = 1 \ \Rightarrow \ \lambda = \pm 1. $$
\end{proof}
\begin{remark}
    When $\lambda = -1$ and $u_i^2 \in \{ -1,+1 \}$, we obtain a Clifford algebra $\mathrm{Cl}_{a,b}$ where $a,b \in \mathbb{N}$ denotes the number of principal units whose square gives $+1$ and $-1$, respectively, with the anticommutation relations $\{ u_i, u_j \} = 0$ for $i \neq j$. However, when $\lambda = +1$ we obtain algebras with the commutation relations $[u_i, u_j] = 0$ for the principal units. The multicomplex number spaces $\mathbb M_n$, see \cite{courchesneAndTremblay2025, GBPrice}, are specific examples of such algebras with $u_i^2 = -1, \forall i = 1,\ldots,n$, since the principal units commute.
\end{remark}
The basis of a hypercomplex algebra consists of all the principal and composite units, that is, all possible combinations, by juxtaposition, of the units $\{ u_i \}$ (including $1$). It can be written simply in terms of the power set $\mathcal{P}(\{ 1, \ldots, n \})$ and is typically called the \textit{standard basis} of the algebra:
$$ \{ u_\mathcal{A} \mid \mathcal{A} \in \mathcal{P}_n \}, \quad \mathcal{P}_n := \mathcal{P}(\{ 1, \ldots, n \}), $$
where the empty set is associated with the index zero and $u_0 := 1$, the singleton $\{k\} \in \mathcal{P}_n$ is associated with the index $k$ and the set $\{ k,l \} \in \mathcal{P}_n$ is associated with the indices $kl$ such that $u_{kl} := u_k u_l$, etc. Another unique representation that will help in the following theorem is by enumerating the basis elements
$$ \mathrm e_0, \mathrm e_1, \ldots, \mathrm e_{2^n-1} $$
using the base-$2$ representation of integers. In each element of the basis, either the factor $u_j$ is present or not, and each can be associated to the $(j-1)$th position (starting from the least significant bit at index $0$) of an integer $0 \leq p < 2^n$. Explicitly,
$$ p = \sum_{k=0}^{n-1} p_k 2^k, \quad p_k \in \{0,1\} $$
and we define
\begin{equation}
    \label{XOR algebra : binary representation}
    \mathrm e_p := \prod_{k=0}^{n-1} u_{k+1}^{p_k}, \quad u_j^0 = 1, \quad j = 1,\ldots,n.
\end{equation}
Clearly this defines the $2^n$ elements of the basis, and this representation is unique (two different integers are associated to different elements of the basis).
\begin{theorem}
    \label{XOR algebra : XOR thm}
    The index function for the standard basis of a hypercomplex algebra is the componentwise XOR function $r(p,q) = p \oplus q$.
\end{theorem}
\begin{proof}
    We have seen that the $2^n$ elements of the basis can be written using (\ref{XOR algebra : binary representation}), so by multiplying two elements $\mathrm e_p$ and $\mathrm e_q$ together we get:
    $$ \begin{aligned}
        \mathrm e_p \mathrm e_q &= \prod_{k=0}^{n-1} u_{k+1}^{p_k} \prod_{l=0}^{n-1} u_{l+1}^{q_l}\\
        &= u_1^{p_0} u_{2}^{p_1} \ldots u_n^{p_{n-1}} u_1^{q_0} u_{2}^{q_1} \ldots u_n^{q_{n-1}}\\
        &= \lambda^{\alpha(p,q)} (u_1^{p_0}u_1^{q_0})(u_2^{p_1}u_2^{q_1}) \ldots (u_n^{p_{n-1}}u_n^{q_{n-1}})\\
        &= \lambda^{\alpha(p,q)} \prod_{k=0}^{n-1} u_{k+1}^{p_k}u_{k+1}^{q_k}
    \end{aligned} $$
    where $\alpha(p, q)$ denotes how many times the commutation property has been applied. For each term of the product, there are four possible cases depending on the values of $p_k, q_k~\in~\{0,1\}$:
    \begin{equation}
        \label{XOR algebra : XOR thm table}
        \begin{array}{c | c | c }
            p_k & q_k & u_{k+1}^{p_k}u_{k+1}^{q_k}\\
            \hline
            0 & 0 & 1\\
            0 & 1 & u_{k+1}\\
            1 & 0 & u_{k+1}\\
            1 & 1 & \beta(k) \in \{ -1, 0 ,1 \}
        \end{array}
    \end{equation}
    So the result depends on the XOR function $p_k \oplus q_k$,
    \begin{equation}
        \label{XOR algebra : XOR thm last equation}
        \mathrm e_p \mathrm e_q = c(p,q)\lambda^{\alpha(p,q)} \prod_{k=0}^{n-1} u_{k+1}^{p_k \oplus q_k} = c(p,q)\lambda^{\alpha(p,q)} \mathrm e_{p \oplus q}
    \end{equation}
    where any extra factor $\beta(k)$ has been absorbed in $c(p,q) \in \{ -1, 0 , +1 \}$. By setting $s(p,q) := c(p,q)\lambda^{\alpha(p,q)}$ , the proof is complete.
\end{proof}
\begin{proposition}
    \label{XOR algebra : values of s}
    For any hypercomplex algebra, the multiplier function for the standard basis has the following properties:
    \begin{enumerate}
        \item $s(p,q) \in \{ -1,0,+1 \}, \forall p,q$;
        \item if $u_i^2 \in \{ -1,+1 \}, \forall i$, then $s(p,q) \in \{ -1, +1 \}, \forall p,q$;
        \item if $u_i^2 =1, \forall i$ and $\lambda=1$, then $s(p,q) =1, \forall p,q$.
    \end{enumerate}
\end{proposition}
\begin{proof}
    From the proof of the theorem \ref{XOR algebra : XOR thm}, in the table (\ref{XOR algebra : XOR thm table}) and equation (\ref{XOR algebra : XOR thm last equation}), we see that the multiplier function is 
    $$ s(p,q) := c(p,q)\lambda^{\alpha(p,q)} $$
    where $c(p,q)$ contains only appearing factors $\beta(k) \in \{ -1, 0 , +1 \}$. Clearly,
    \begin{enumerate}
        \item $c(p,q) \in \{ -1, 0 , +1 \}$, and
        \item if $u_i^2 \in \{ -1 ,1 \}, \forall i$, then $c(p,q) \in \{ -1 , +1 \}$.
    \end{enumerate}
    From proposition \ref{XOR algebra : commutation property lambda values}, $\lambda = \pm 1 \Rightarrow \lambda^{\alpha(p,q)} \in \{ -1,+1 \}$. Then $s(p,q) \in \{ -1, 0, +1 \}$ and $s(p,q) = 0$ if and only if $c(p,q) = 0$. For the third case, if $u_i^2 = 1, \forall i$ and $\lambda=1$, then $s(p,q) = c(p,q)\lambda^{\alpha(p,q)} = 1, \forall p,q$.
\end{proof}

\section{Multiplicative diagonal basis}
\label{section : multiplicative diagonal basis}

We may be interested in a general expression for the multiplier and index functions when a linear transformation is applied to obtain a new basis. The result of such an operation can be quite complicated, and does not guarantee the existence of a multiplier and index function in the new basis obtained. However, we can make some assumptions depending on what we are looking for. Something interesting and generally useful to simplify computations is a multiplicative diagonal (or canonical) basis.
\begin{definition}
    \label{multiplicative diagonal basis : diagonal basis def}
    A basis $\{ {\tilde e}_{\alpha} \}_{\alpha=0}^{n-1}$ of $V$ is said to be diagonal or canonical if
    $$ {\tilde e}_{\alpha} * {\tilde e}_{\beta} = \tilde s(\alpha,\beta) {\tilde e}_{\tilde r (\alpha,\beta)}, $$
    $$ \tilde s(\alpha,\beta) = \delta_{\alpha,\beta} \quad \text{and} \quad \tilde r(\alpha,\beta) = \alpha \delta_{\alpha,\beta}, \quad \alpha,\beta = 0,\ldots,n-1, $$
    where $\delta_{\alpha,\beta}$ is the Kronecker delta.
\end{definition}
An immediate result of this definition is the following proposition, since ${\tilde e}_{\alpha} * {\tilde e}_{\beta} = 0$ for all $\alpha \neq \beta$ if such a diagonal basis exists.
\begin{proposition}
    If $V$ is an algebra of dimension $n \geq 2$ and has a diagonal basis, then all elements of this basis are zero divisors.
\end{proposition}
\begin{theorem}
    \label{multiplicative diagonal basis : basic thm}
    Let $V$ be an algebra over $\mathbb{K}$ of dimension $n$ and denotes its bilinear binary operation by $*$. Suppose $V$ has a basis $\{ \mathrm e_k \}_{k=0}^{n-1}$ satisfying
    $$ \mathrm e_p * \mathrm e_q = s(p,q) \mathrm e_{r(p,q)}, \quad p,q = 0,\ldots,n-1 $$
    for some functions $s:\{0,\ldots,n-1\}^2 \rightarrow \mathbb{K}$ and $r:\{0,\ldots,n-1\}^2 \rightarrow \{0,\ldots,n-1\}$ and let $\{ {\tilde e}_{\alpha} \}_{\alpha=0}^{n-1}$ be a diagonal basis, related to the first one by the change-of-coordinates matrix $T$ \cite{sergeLangLinear, Lipschutz}:
    $$ \mathrm e_k = \sum_{\alpha=0}^{n-1} T_{\alpha, k} {\tilde e}_{\alpha}, \quad k = 0,\ldots, n-1. $$
    Then
    $$ s(p,q) T_{\alpha, r(p,q)} = T_{\alpha, p} T_{\alpha, q}, \quad \alpha,p,q=0,\ldots,n-1. $$
\end{theorem}
\begin{proof}
    Let's apply the binary operation $*$ on an arbitrary pair of basis elements $\mathrm e_p, \mathrm e_q$:
    \begin{equation} 
        \label{multiplicative diagonal basis : diagonal basis thm first part}
        \begin{aligned}
            \mathrm e_p * \mathrm e_q &= \sum_{\alpha,\beta=0}^{n-1} T_{\alpha, p} T_{\beta, q} ({\tilde e}_{\alpha} * {\tilde e}_{\beta})\\
            &= \sum_{\alpha,\beta=0}^{n-1} T_{\alpha, p} T_{\beta, q} \delta_{\alpha\beta} {\tilde e}_{\alpha \delta(\alpha,\beta)}\\
            &= \sum_{\alpha=0}^{n-1} T_{\alpha, p} T_{\alpha, q} {\tilde e}_{\alpha}.
        \end{aligned} 
    \end{equation}
    The left-hand side can be expanded into 
    \begin{equation} 
        \label{multiplicative diagonal basis : diagonal basis thm second part}
        \mathrm e_p * \mathrm e_q = s(p,q) \mathrm e_{r(p,q)} = \sum_{\alpha=0}^{n-1} s(p,q) T_{\alpha, r(p,q)} {\tilde e}_{\alpha}.
    \end{equation}
    From the linear independence of the basis elements ${\tilde e}_{\alpha}, \alpha=0,\ldots,n-1$, combining (\ref{multiplicative diagonal basis : diagonal basis thm first part}) and (\ref{multiplicative diagonal basis : diagonal basis thm second part}) implies
    $$ s(p,q) T_{\alpha, r(p,q)} = T_{\alpha, p} T_{\alpha, q}, \quad \alpha,p,q=0,\ldots,n-1. $$
\end{proof}

\subsection{Diagonal basis from the XOR index function}
\label{subsection : diagonal basis for the XOR index function}

As we saw in section \ref{subsection : some examples} and \ref{section : Commutation property and XOR index function}, the XOR function is a recurring pattern and at the heart of many algebra and hypercomplex number spaces. So we now specialize to the case where $V$ is a vector space of dimension $2^n$ equipped with the bilinear operation $*$, and where the index function of the basis $\{ \mathrm e_k \}_{k=0}^{2^n-1}$ is $r(p,q) = p \oplus q$, with no restriction on the multiplier function $s(p,q)$. For the following, we also suppose that there is a diagonal basis $\{ {\tilde e}_{\alpha} \}_{\alpha=0}^{2^n-1}$ related to the initial basis through the change-of-coordinates matrix $T$, so that theorem \ref{multiplicative diagonal basis : basic thm} applies. This set of assumptions will lead to the necessary conditions to get such a basis from the XOR index function.
\begin{lemma}
    \label{diagonal basis XOR : T nonzero}
    All entries of the change-of-coordinates matrix $T$ are nonzero i.e $T_{\alpha,p} \neq 0$ for all $\alpha,p = 0,\ldots, 2^n-1$.
\end{lemma}
\begin{proof}
    Notice that $r(p,0) = p \oplus 0 = p$, so from theorem \ref{multiplicative diagonal basis : basic thm},
    \begin{equation}
        \label{diagonal basis XOR : lemma T q zero}
        s(p,0) T_{\alpha, p} = T_{\alpha, p} T_{\alpha, 0}.
    \end{equation}
    For any $\alpha \in \{ 0,\ldots,2^n-1 \}$, we can find at least one $p'_\alpha \in \{ 0,\ldots, 2^n-1 \}$ such that $T_{\alpha, p'_\alpha} \neq 0$, unless the whole $\alpha$th row of $T$ contains zeros, which is a contradiction ($T$ is invertible by definition). Then from (\ref{diagonal basis XOR : lemma T q zero})
    \begin{equation}
        \label{diagonal basis XOR : lemma T nonzero free row}
        T_{\alpha, 0} = s(p'_\alpha,0).
    \end{equation}
    In the same way, we can find
    \begin{equation}
        \label{diagonal basis XOR : lemma T p zero}
        s(0,q) T_{\alpha, q} = T_{\alpha, 0} T_{\alpha, q},
    \end{equation}
    \begin{equation}
        \label{diagonal basis XOR : lemma T nonzero free col}
        T_{\alpha, 0} = s(0,q_{\alpha}'),
    \end{equation}
    for values $q_{\alpha}' \in \{ 0,\ldots,2^m-1 \}$ such that $T_{\alpha,q_{\alpha}'} \neq 0$. A second useful property of the XOR is that $r(p,p) = p \oplus p = 0$ for any value of $p$, so from theorem \ref{multiplicative diagonal basis : basic thm}, equations (\ref{diagonal basis XOR : lemma T nonzero free row}) and (\ref{diagonal basis XOR : lemma T nonzero free col}),
    \begin{equation}
        \label{diagonal basis XOR : lemma T q is p}
        T_{\alpha,p}^2 = s(p,p) T_{\alpha,0} = s(p,p) s(p'_\alpha,0) = s(p,p) s(0,q_{\alpha}').
    \end{equation}
    From the definition of the values $p'_{\alpha}$ above, $T_{\alpha, p'_\alpha} \neq 0$ which implies from (\ref{diagonal basis XOR : lemma T q is p}),
    $$ s(p'_\alpha,p'_\alpha) s(p'_\alpha,0) = T^2_{\alpha, p'_\alpha} \neq 0 $$
    \begin{equation}
        \label{diagonal basis XOR : lemma T s nonzero}
        \Rightarrow \quad s(p_{\alpha}', p_{\alpha}') \neq 0 \quad \text{and} \quad s(p_{\alpha}',0) \neq 0.
    \end{equation}
    The last part of the proof is done by contradiction. Suppose there is a pair of values $\alpha_0, p_0$ such that $T_{\alpha_0, p_0} = 0$. From (\ref{diagonal basis XOR : lemma T q is p}) and (\ref{diagonal basis XOR : lemma T s nonzero}),
    $$ 0 = s(p_0,p_0) s(p'_{\alpha_0},0) \Rightarrow s(p_0,p_0) = 0. $$
    Then for any $\alpha = 0,\ldots,2^n-1$,
    $$ T_{\alpha,p_0}^2 = s(p_0,p_0) s(p'_\alpha,0) = 0 $$
    which means the whole $p_0$th column of $T$ is made of zeros, again a contradiction since $T$ is invertible. So no such pair $\alpha_0, p_0$ of values exists, and $T_{\alpha,p} \neq 0$ for all $\alpha,p = 0, \ldots, 2^n-1$.
\end{proof}
\begin{theorem}
    For all $\alpha, p, q = 0,\ldots, 2^n-1$,
    \begin{enumerate}
        \item $T_{\alpha,0} = s(0,0)$,
        \item $T_{\alpha,p}^2 = s(p,p) s(0,0)$,
        \item $s(p,q) = s(q,p)$,
        \item $s(p,p) \neq 0$,
        \item $s(p,0) = s(0,q) = s(0,0)$,
        \item $s(p,p) s(0,0)$ has a square root relative to $*$, with at least two distinct solutions in $\mathbb{K}$.
    \end{enumerate}
\end{theorem}
\begin{proof}
    From $r(p,0) = p$, $r(0,q) = q$, theorem \ref{multiplicative diagonal basis : basic thm} and lemma \ref{diagonal basis XOR : T nonzero},
    \begin{equation}
        \label{diagonal basis XOR : T thm first part}
        T_{\alpha,0} = s(p,0) = s(0,q)
    \end{equation}
    which is true for all $\alpha,p,q = 0,\ldots,2^n-1$ and proves both conditions 1. and 5. Using $r(p,p) = 0$ and (\ref{diagonal basis XOR : T thm first part}),
    $$ s(p,p) T_{\alpha,0} = T_{\alpha,p}^2 \Rightarrow T_{\alpha,p}^2 = s(p,p) s(0,0). $$
    From lemma \ref{diagonal basis XOR : T nonzero}, $T_{\alpha,p} \neq 0$ so $s(p,p) \neq 0$ as well for all values of $p$, which validates both conditions 2. and 4. For condition 3., we use the fact that the XOR function is commutative $r(p,q) = r(q,p)$, so we get
    $$ s(p,q) T_{\alpha, r(p,q)} = T_{\alpha,p} T_{\alpha,q} = s(q,p) T_{\alpha, r(q,p)} = s(q,p) T_{\alpha, r(p,q)} $$
    and $T_{\alpha, r(p,q)} \neq 0$, implying that $s(p,q) = s(q,p)$. Lastly, suppose that there is only a single solution $c(p) \in \mathbb{K}$ such that
    $$ T_{\alpha,p} = \sqrt{s(p,p) s(0,0)} = c(p) $$
    then the $p$th column of $T$ is collinear with the $0$th column $T_{\alpha,0} = s(0,0)$, which is a contradiction since $T$ is a change-of-coordinates matrix. This implies that at least two solutions for the square root of $s(p,p) s(0,0)$ must exist in order to get linearly independent columns.
\end{proof}
\begin{remark}
    This theorem gives both the equations to build the change-of-coordinates matrix from the multiplier function and a set of necessary conditions (constraints) on $s(p,q)$ for the diagonal basis to exist. It is, however, not a set of sufficient conditions and cannot assert by itself the existence of such a basis.
\end{remark}
The next corollary follows directly from the previous theorem, since if $\mathrm e_0^2 = \mathrm e_0$, then $s(0,0) = 1$, which is also true if $\mathrm e_0$ is the multiplicative neutral element.
\begin{corollary}
    \label{XOR diagonal basis : corollary}
    If $\mathrm e_0^2 = \mathrm e_0$ from the initial basis or if $\mathrm e_0$ is a neutral element relative to $*$, then for all $\alpha, p, q = 0,\ldots, 2^n-1$,
    \begin{enumerate}
        \item $T_{\alpha,0} = s(0,0) = 1$,
        \item $T_{\alpha,p}^2 = s(p,p)$,
        \item $s(p,q) = s(q,p)$,
        \item $s(p,p) \neq 0$,
        \item $s(p,0) = s(0,q) = 1$,
        \item $s(p,p)$ has a square root relative to $*$, with at least two distinct solutions in $\mathbb{K}$.
    \end{enumerate}
\end{corollary}

\sloppy\section{Diagonal basis in hypercomplex algebras}
\label{Diagonal basis in hypercomplex algebras}

Using the tools of section \ref{subsection : diagonal basis for the XOR index function}, we can now find which hypercomplex algebras (definition \ref{XOR algebra : hypercomplex algebra def}) may have a diagonal basis or not. The Sylvester-Hadamard matrices will be introduced later for a direct algorithm-like method to find the change-of-coordinates matrix, linking the standard basis to the diagonal basis, in specific hypercomplex algebras in which the necessary conditions of corollary \ref{XOR diagonal basis : corollary} are met. This proves at the same time the converse of the corollary by ensuring the existence of the diagonal basis for these particular cases.

Conditions 3. and 4. of corollary \ref{XOR diagonal basis : corollary} greatly restricts the class of algebras containing a diagonal basis, since any with a principal unit $u_i^2 = 0$ is excluded, as well as any that is not commutative. For this reason, we restrict to the cases where $u_i^2 \in \{ -1 ,+1 \}$. Proposition \ref{XOR algebra : commutation property lambda values} then separates, by setting $\lambda = -1$ or $\lambda = +1$, the hypercomplex algebras into two distinct subclasses : Clifford algebras with anticommutative units and another subclass with commutative units:
$$ \begin{aligned}
    & \lambda = -1 \Rightarrow u_i u_j = -u_j u_i, \quad i \neq j,\\
    & \lambda = +1 \Rightarrow u_i u_j = u_j u_i.
\end{aligned} $$
We will see later how this subclass is related to multiperplex and multicomplex numbers.
\begin{proposition}
    The only Clifford algebras with a diagonal basis are:
    \begin{enumerate}
        \item $\mathrm{Cl}_{0,0}(\mathbb{R}) \simeq \mathbb{R}$,
        \item $\mathrm{Cl}_{1,0}(\mathbb{R}) \simeq \mathbb{D}$ (hyperbolic numbers),
        \item $\mathrm{Cl}_{0,0}(\mathbb{C}) \simeq \mathbb{C}$,
        \item $\mathrm{Cl}_{1,0}(\mathbb{C}) \simeq \mathbb{M}_2(\mathbb{R})$ (bicomplex numbers),
        \item $\mathrm{Cl}_{0,1}(\mathbb{C}) \simeq \mathbb{M}_2(\mathbb{R})$ (bicomplex numbers).
    \end{enumerate}
\end{proposition}
\begin{proof}
    If the number of principal units is greater or equal than $2$, then
    $$ u_1 u_2 = - u_2 u_1 $$
    so $s(1,2) = -s(2,1)$ and the resulting algebra does not contain a diagonal basis. Cases 1. and 3. are trivial, while the existence of the diagonal basis in the hyperbolic and bicomplex number spaces is already known \cite{GuillaumeBrouillette, courchesneAndTremblay2025}. The remaining case is $\mathrm{Cl}_{0,1}(\mathbb{R})$, however in this algebra
    $$ u_1^2 = - 1 \Rightarrow s(1,1) = -1 $$
    which doesn't have a square root in $\mathbb{R}$. By corollary \ref{XOR diagonal basis : corollary}, there is no diagonal basis in $\mathrm{Cl}_{0,1}(\mathbb{R})$.
\end{proof}
\begin{definition}
    \label{Diagonal basis in hypercomplex algebras : definition pre-multicomplex/multiperplex}
    We denote by $\mathbb{M}_{a,b}(\mathbb{K}), a,b \in \mathbb{N}$ the hypercomplex algebra over a field $\mathbb{K}$ with commutative units $(\lambda = +1)$, such that
    $$ u_1^2 = \ldots = u_a^2 = +1 \quad \text{and} \quad u_{a+1}^2 = \ldots = u_{a+b}^2 = -1. $$
\end{definition}
\begin{theorem}
    \label{Diagonal basis in hypercomplex algebras : thm equivalence of multicomplex spaces}
    Let $a,b \in \mathbb{N}$. Then we have the following algebra isomorphism
    $$ \mathbb{M}_{a,b}(\mathbb{C}) \simeq \mathbb{M}_{a,b+1}(\mathbb{R}) \simeq \mathbb{M}_{0,a+b+1}(\mathbb{R}) \simeq \mathbb{M}_{0,a+b}(\mathbb{C}). $$
\end{theorem}
\begin{proof}
    The first isomorphism $\mathbb{M}_{a,b}(\mathbb{C}) \simeq \mathbb{M}_{a,b+1}(\mathbb{R})$ is obtained directly by denoting the standard basis of $\mathbb{M}_{a,b}(\mathbb{C})$ as $\{ \mathrm e_k \}_{k=0}^{2^n-1}$ and the basis of $\mathbb{M}_{a,b+1}(\mathbb{R})$ as $\{ \mathrm e_k \}_{k=0}^{2^n-1} \cup \{ i\mathrm e_k \}_{k=0}^{2^n-1}$ where $i$ is the extra unit such that $i^2 = -1$. Let $z \in \mathbb{M}_{a,b}(\mathbb{C})$, we can write
    $$ z = \sum_{k=0}^{2^n-1} z_k \mathrm e_k = \sum_{k=0}^{2^n-1} (x_k + i y_k) \mathrm e_k = \sum_{k=0}^{2^n-1} x_k \mathrm e_k +  \sum_{k=0}^{2^n-1} y_k (i\mathrm e_k) \in \mathbb{M}_{a,b+1}(\mathbb{R}) $$
    where $x_k, y_k \in \mathbb{R}$. Since $i$ commutes with all the other units, both algebras have the same multiplication table, both on $\mathbb{R}$ and $\mathbb{C}$, so they are isomorphic.

    The second isomorphism $\mathbb{M}_{a,b+1}(\mathbb{R}) \simeq \mathbb{M}_{0,a+b+1}(\mathbb{R})$ is more subtle and requires the following binomial identity \cite{GrahamKnuthPatashnik} for $m \geq 1$
    \begin{equation}
        \label{even odd binomial identities}
        2^{m-1} = \sum_{l=0}^{\lfloor m/2 \rfloor} \binom{m}{2l} = \sum_{l=0}^{\lfloor m/2 \rfloor-1} \binom{m}{2l + 1}.
    \end{equation}
    Let $b' := b+1$ and $n := a + b'$. The principal units of $\mathbb{M}_{a,b'}(\mathbb{R})$ satisfy
    $$ u_1^2 = \ldots = u_a^2 = +1, \quad u_{a+1}^2 = \ldots = u_{a + b'}^2 = -1 $$
    and from equation (\ref{XOR algebra : binary representation}), any element $\mathrm e_p$ of the standard basis can be written
    $$ \mathrm e_p = \prod_{k=0}^{n-1} u_{k+1}^{p_k}, \quad \text{with} \quad p = (p_{n-1}\ldots p_1 p_0)_2, \quad p_k \in \{ 0,1 \}. $$
    We will now prove that the number of elements from the standard basis that square to $-1$ or $+1$ does not depend on the values of $a$ and $b'$ when $b' \geq 1$. Since all the units commute, if $a \geq 1$,
    $$ \mathrm e_p^2 = \prod_{k=0}^{a-1} u_{k+1}^{2p_k} \prod_{k=a}^{n-1} u_{k+1}^{2p_k} = \prod_{k=0}^{a-1} (+1)^{p_k} \prod_{k=a}^{n-1} (-1)^{p_k}. $$
    The value $\mathrm e_p^2 \in \{ +1, -1 \}$ depends on the number of bits in $p$ that are equal to $1$ for $a \leq k \leq n-1$. For an even number of bits, the value is $+1$, for an odd number, the result is $-1$. The first $a$ bits do not affect this value, but still add a factor $2^a$ to the number of possible configurations. This is where the identity (\ref{even odd binomial identities}) becomes handy: the number of configurations for the last $b'$ bits with an even or odd numbers of bits equal to $1$ is $2^{b'-1}$, so we get
    $$ \begin{aligned}
        |\{ p | \mathrm e_p^2=-1 \}| &= 2^a 2^{b'-1} = 2^{n-1},\\
        |\{ p | \mathrm e_p^2=+1 \}| &= 2^a 2^{b'-1} = 2^{n-1}.\\
    \end{aligned} $$
    If $a = 0$, we can see from 
    $$ \mathrm e_p^2 = \prod_{k=0}^{n-1} u_{k+1}^{2p_k} = \prod_{k=0}^{n-1} (-1)^{p_k} $$
    that the result is still the same. Since the algebras have the same dimension and the same multiplication table for any values of $a$ and $b$, we get $\mathbb{M}_{a,b+1}(\mathbb{R}) \simeq \mathbb{M}_{0,a+b+1}(\mathbb{R})$.

    The third isomorphism is obtained directly by using the first one.
\end{proof}
Theorem \ref{Diagonal basis in hypercomplex algebras : thm equivalence of multicomplex spaces} makes us notice a redundancy in the notation $\mathbb{M}_{a,b}(\mathbb{K})$, since there are only two distinct cases: $\mathbb{M}_{n,0}(\mathbb{R})$ and $\mathbb{M}_{n,0}(\mathbb{C})$, $n \in \mathbb{N}$, respectively known as the multiperplex and multicomplex number spaces \cite{courchesneAndTremblay2025}. This conclusion is consistent with recent results from H. Sharma and D. Shirikov \cite{SharmaCommutativeAnalogs}. Both spaces are related via a complexification operation and have the same dimension $2^n$ over their respective field:
$$ \mathbb{M}_{n,0}(\mathbb{R}) \rightarrow \mathbb{M}_{n,0}(\mathbb{C}) \simeq \mathbb{M}_{0,n}(\mathbb{C}). $$
\begin{definition}
    The multiperplex and multicomplex number spaces will be denoted and defined as
    \begin{enumerate}
        \item[$\bullet$] $\mathbb{D}_n := \mathbb{M}_{n,0}(\mathbb{R})$ (multiperplex),
        \item[$\bullet$] $\mathbb{M}_n := \mathbb{M}_{n,0}(\mathbb{C})$ (multicomplex).
    \end{enumerate}
\end{definition}

\subsection{Sylvester-Hadamard matrices}
\label{subsection : Sylvester-Hadamard matrices}

A Hadamard matrix is a square $m \times m$ matrix $H$ with entries in $\{ +1, -1 \}$ such that
\begin{equation}
    \label{def Hadamard matrices}
    H H^{\top} = m I_m
\end{equation}
where $I_m$ is the identity. Here we are mostly interested in the Sylvester construction of Hadamard matrices. Let $\langle p, q\rangle$ denotes the binary inner product of the positive integers $p$ and $q$ i.e. if $\mathrm{bin}_n(p) = (p_{n-1}, p_{n-2},\ldots,p_0)$ and $\mathrm{bin}_n(q) = (q_{n-1}, q_{n-2},\ldots,q_0)$ with $p_j, q_j \in \{ 0,1 \}$, then
$$ \langle p,q\rangle := \sum_{j=0}^{n-1} p_j \wedge q_j, $$
which is the bit count of the AND operation $p \wedge q$. The Sylvester-Hadamard $2^n \times 2^n$ square matrices are defined as
\begin{equation}
    \label{def Sylvester Hadamard matrices}
    H_n := [(-1)^{\langle p,q\rangle}]_{p,q=0}^{2^n-1}, \quad n \geq 1.
\end{equation}
We already know from \cite{horadam} that $H_n$ satisfy (\ref{def Hadamard matrices}) with $m=2^n$, it is also this book that inspired this present notation. By such, $H_n$ is invertible, and also commutative by definition ($\langle p, q\rangle = \langle q, p \rangle$):
$$ H_n^{-1} = \frac{H_n}{2^n}, \quad H_n^{\top} = H_n. $$
\begin{proposition}
    \label{bit scalar product property}
    Let $p,q \in \mathbb{N}$ such that $0 \leq p,q < 2^n$, then
    $$ \sum_{k=0}^{2^n-1} (-1)^{\langle p,k \rangle} (-1)^{\langle k,q \rangle} = 2^n \delta_{p,q} $$
    where $\delta_{p,q}$ is the Kronecker delta.
\end{proposition}
\begin{proof}
    This statement is equivalent to definition (\ref{def Sylvester Hadamard matrices}) and equation (\ref{def Hadamard matrices}). Here we prove it independently for another application case. From the definition,
    $$ \langle p,k \rangle + \langle k,q \rangle = \sum_{j=0}^{n-1} k_j \wedge p_j + \sum_{j=0}^{n-1} k_j \wedge q_j = \sum_{j=0}^{n-1} (k_j \wedge p_j + k_j \wedge q_j) $$
    and it can be shown using a truth table that
    $$ \alpha_j := k_j \wedge p_j + k_j \wedge q_j = k_j \wedge (p_j \vee q_j) + k_j \wedge p_j \wedge q_j. $$
    From this equality, we can see that
    $$ \alpha_j = \begin{cases}
        0 & \text{if } p_j = q_j = 0,\\
        2 k_j & \text{if } p_j = q_j = 1,\\
        k_j & \text{if } p_j \neq q_j.
    \end{cases} $$
    If $p=q$, 
    $$ \sum_{k=0}^{2^n-1} (-1)^{\langle p,k \rangle} (-1)^{\langle k,p \rangle} = \sum_{k=0}^{2^n-1} (-1)^{2 \langle k,p \rangle} = 2^n. $$
    If $p \neq q$, we denote by $j_1, \ldots, j_m$ all the indices such that $p_j \neq q_j$. Note that $\alpha_j$ is even for all $p_j = q_j$, and $\alpha_j = k_j$ when $p_j \neq q_j$, so
    $$ \sum_{k=0}^{2^n-1} (-1)^{\sum_j \alpha_j} = \sum_{k=0}^{2^n-1} (-1)^{k_{j_1}} \ldots (-1)^{k_{j_m}}. $$
    The sum over $k$ can be written as the sum of all possible bit values $k_j = 0,1$:
    $$ \begin{aligned}
        \sum_{k=0}^{2^n-1} (-1)^{k_{j_1}} \ldots (-1)^{k_{j_m}} &= \sum_{k_0 = 0}^1 \ldots \sum_{k_{n-1} = 0}^1 (-1)^{k_{j_1}} \ldots (-1)^{k_{j_m}}\\
        &= C \sum_{k_{j_1} = 0}^1 (-1)^{k_{j_1}}\\
        &= C (1 - 1) = 0.
    \end{aligned} $$
    So at least one different bit $p_j \neq q_j$, which always exists when $p \neq q$, implies
    $$ \sum_{k=0}^{2^n-1} (-1)^{\langle p,k \rangle} (-1)^{\langle k,q \rangle} = 0. $$
\end{proof}
\begin{proposition}
    \label{hadamard-solution matrix def and properties}
    Let $T_n := [\nu_q (-1)^{\langle p,q\rangle}]_{p,q=0}^{2^n-1}$ where $\nu_q \in \mathbb{C}$ and $|\nu_q| = 1$. Then
    $$ T_n^{\dagger} T_n = 2^n I_{2^n}. $$
\end{proposition}
\begin{proof}
    The equality $T_n^{\dagger} T_n = 2^n I_{2^n}$ is equivalent to
    $$ (T_n^{\dagger} T_n)_{p,q} = 2^n \delta_{p,q}. $$
    Expanding the matrix product,
    $$ \begin{aligned}
        (T_n^{\dagger} T_n)_{p,q} &= \sum_{k=0}^{2^n-1} (T_n^{\dagger})_{p,k} (T_n)_{k,q}\\
        &= \sum_{k=0}^{2^n-1} \nu_p^* \nu_q (-1)^{\langle p,k\rangle} (-1)^{\langle k,q\rangle}\\
        &= \nu_p^* \nu_q \sum_{k=0}^{2^n-1} (-1)^{\langle p,k\rangle} (-1)^{\langle k,q\rangle}\\
        &= |\nu_p|^2 2^n \delta_{p,q} \quad \text{(proposition \ref{bit scalar product property})}\\
        &= 2^n \delta_{p,q}.
    \end{aligned} $$
\end{proof}
\begin{theorem}
    \label{theorem : algorithm-like diagonal basis}
    Let $s$ be the multiplier function of a commutative hypercomplex algebra $\mathcal{A}$ with $n \geq 0$ principal units $\{1,u_1, \ldots,u_n\}$ such that $u_i^2 \neq 0, \forall i$. Then there is a multiplicative diagonal basis given by the change-of-coordinates matrix
    $$ T_n := [\nu_q (-1)^{\langle p,q\rangle}]_{p,q=0}^{2^n-1}, $$
    from the standard basis to the new basis, with $\nu_q^2 = s(q,q)$.
\end{theorem}
\begin{proof}
    Let's denote by $\mathrm e_p, p=0,\ldots,2^n-1$ the elements of the standard basis as defined in equation (\ref{XOR algebra : binary representation}). From proposition \ref{XOR algebra : values of s}, $s(p,q) \in \{ -1, +1 \}$ since $u_k^2 \in \{ -1, +1 \}$ which implies that we can choose $\nu_q \in \{ 1,i \}$ with $\nu_q^2 = 1^2 = 1$ for the cases where $s(q,q) = 1$, and $\nu_q^2 = i^2 = -1$ for the cases where $s(q,q)=-1$. 
    
    
    We also denote by $\tilde{\mathrm{e}}_p, p=0,\ldots,2^n-1$ the elements of the new basis using the change-of-coordinates matrix $T_n$:
    $$ \tilde{\mathrm e}_k := \sum_{\alpha=0}^{2^n-1} (T_n)^{-1}_{\alpha,k} \mathrm e_\alpha. $$
    Since $|\nu_q| = 1$, we know from proposition \ref{hadamard-solution matrix def and properties} that the inverse of $T_n$ exists and is given by
    $$ (T_n)^{-1} = \frac{T_n^{\dagger}}{2^n}. $$
    The multiplier $\tilde{s}$ and the index function $\tilde{r}$ of the new basis are defined such that
    $$ \tilde{\mathrm e}_p \tilde{\mathrm e}_q = \tilde s(p,q) \tilde{\mathrm e}_{\tilde r(p,q)} $$
    and this basis is diagonal if $\tilde{s}(p,q) = \delta_{p,q}$ and $\tilde{r}(p,q) = p \delta_{p,q}$. We then proceed by induction, noting that the diagonal basis trivially exists when $n=0$. Suppose that the diagonal basis $\{ \hat{\mathrm e}_p \}_{p=0}^{2^n-1}$ exists for an algebra with $n$ principal units, obtained with the matrix $T_n$, and we add the unit $u_{n+1}$ such that $u_{n+1}^2 \in \{ -1, +1 \}$. Then for $p,q < 2^n$,
    \begin{equation}
        \label{multiperplex and multicomplex diagonal basis theorem : elements e}
        \mathrm e_{p+2^n} = \prod_{k=0}^{n} u_{k+1}^{(p+2^n)_k} = \prod_{k=0}^{n-1} u_{k+1}^{p_k} \cdot u_{n+1} = u_{n+1} \mathrm e_p.
    \end{equation}
    By definition of the bitwise scalar product,
    \begin{equation}
        \label{multiperplex and multicomplex diagonal basis theorem : bitwise scalar product}
        \begin{aligned}
            & \langle p + 2^n,q\rangle = \langle p,q + 2^n\rangle =  \langle p,q\rangle,\\
            & \langle p + 2^n,q + 2^n\rangle = \langle p,q\rangle + 1.
        \end{aligned}
    \end{equation}
    And by definition of $\nu_q$,
    $$ s(q+2^n, q+2^n) = \mathrm e_{q+2^n} \mathrm e_{q+2^n} = u_{n+1}^2 \mathrm e_q \mathrm e_q = s(2^n, 2^n) s(q,q) $$
    \begin{equation}
        \label{multiperplex and multicomplex diagonal basis theorem : elements nu}
        \Rightarrow \nu_{q+2^n} = \sqrt{s(q+2^n, q+2^n)} = \sqrt{s(2^n, 2^n) s(q,q)} = \nu_{2^n} \nu_q.
    \end{equation}
    From these last equations, we get
    $$
        \begin{aligned}
            (T_{n+1})_{p,q} &= \nu_q (-1)^{\langle p,q\rangle} = (T_n)_{p,q},\\
            (T_{n+1})_{p+2^n,q} &= \nu_q (-1)^{\langle p+2^n,q\rangle} = \nu_q (-1)^{\langle p,q\rangle} = (T_n)_{p,q},\\
            (T_{n+1})_{p,q+2^n} &= \nu_{q+2^n} (-1)^{\langle p,q\rangle} = \nu_{2^n}\nu_q (-1)^{\langle p,q\rangle} = \nu_{2^n}(T_n)_{p,q},\\
            (T_{n+1})_{p+2^n,q+2^n} &= \nu_{q+2^n} (-1)^{\langle p+2^n,q+2^n\rangle} = -\nu_{2^n}\nu_q (-1)^{\langle p,q\rangle} = -\nu_{2^n}(T_n)_{p,q},
        \end{aligned}
    $$
    or in the form of a block matrix
    \begin{equation}
        T_{n+1} = \begin{bmatrix}
            T_n & \nu_{2^n} T_n\\
            T_n & - \nu_{2^n} T_n
        \end{bmatrix} \Rightarrow (T_{n+1})^{-1} = \frac{1}{2} \begin{bmatrix}
            (T_n)^{-1} & (T_n)^{-1}\\
            \nu^*_{2^n} (T_n)^{-1} & - \nu^*_{2^n} (T_n)^{-1}
        \end{bmatrix}.
    \end{equation}
    The elements of the new basis are given, for $p<2^n$, by
    $$ \begin{aligned}
        \tilde{\mathrm e}_p &= \sum_{\alpha=0}^{2^{n+1}-1} (T_{n+1})^{-1}_{\alpha,p} \mathrm e_\alpha\\
        &= \sum_{\alpha=0}^{2^n-1} \left[ (T_{n+1})^{-1}_{\alpha,p} \mathrm e_\alpha+ (T_{n+1})^{-1}_{\alpha+2^n,p} \mathrm e_{\alpha+2^n} \right]\\
        &= \frac{1}{2} \sum_{\alpha=0}^{2^n-1} \left[ (T_{n})^{-1}_{\alpha,p} \mathrm e_\alpha+ \nu_{2^n}^* (T_{n})^{-1}_{\alpha,p} u_{n+1}\mathrm e_{\alpha} \right]\\
        &= \hat{\mathrm e}_p \frac{(1 + \nu^*_{2^n} u_{n+1})}{2}.
    \end{aligned} $$
    and in the same way,
    $$ \tilde{\mathrm e}_{p + 2^n} = \hat{\mathrm e}_p \frac{(1 - \nu^*_{2^n} u_{n+1})}{2}. $$
        From the following properties,
        $$ \left( \frac{1 + \nu^*_{2^n} u_{n+1}}{2} \right)^2 = \frac{(1 + \nu^*_{2^n} u_{n+1})}{2}, $$
        $$ \left( \frac{1 - \nu^*_{2^n} u_{n+1}}{2} \right)^2 = \frac{(1 - \nu^*_{2^n} u_{n+1})}{2}, $$
        $$ \frac{(1 + \nu^*_{2^n} u_{n+1})}{2} \cdot \frac{(1 - \nu^*_{2^n} u_{n+1})}{2} = 0, $$
        and the fact that $\{ \hat{\mathrm e}_p \}_{p=0}^{2^n-1}$ is a diagonal basis, it is clear that $\{ \tilde{\mathrm e}_k \}_{k=0}^{2^{n+1}-1}$ is a diagonal basis of the new algebra obtained by adding the unit $u_{n+1}$.
\end{proof}
\begin{remark}
    Note that since the only commutative hypercomplex algebras with nonvanishing principal units are the multiperplex $\mathbb{D}_n$ and multicomplex $\mathbb{M}_n$ number spaces defined above to have only hyperbolic units, a simpler proof would have been possible by setting $\nu_q = 1$ everywhere and using directly the Sylvester-Hadamard matrices. We have performed the proof here for the general case with the aim of obtaining an expression for the change-of-coordinates matrix even when defining the algebra using imaginary units.
\end{remark}

\sloppy\section{Conjugate algebra}

Let $\mathcal{A}$ be a commutative hypercomplex algebra as defined in sections \ref{section : Commutation property and XOR index function} and \ref{Diagonal basis in hypercomplex algebras} with a set of principal units $\{1,u_1,\ldots,u_n\}$, $u_i^2 \neq 0, \forall i$. We also assume that the power $x^a$ is defined for all invertible elements $x \in \mathcal{A}$ with $a \in \mathbb{R}$, and non-invertible elements $x \in \mathcal{A}$ with $a \geq 0$. From the commutative property of the algebra, $x^{a+b} = x^ax^b = x^bx^a$ for all $a,b \in \mathbb{R}$ and $x \in \mathcal{A}$ when this expression makes sense.
We define the conjugate operation $\dagger_k$ on $\mathcal{A}$ such that
$$ \dagger_0 := \mathrm{id} \quad \text{and} \quad \dagger_k : u_k \rightarrow -u_k, \quad k=1,\ldots,n $$
and their compositions as
$$ x^{\dagger_j \circ \dagger_k} := (x^{\dagger_j})^{\dagger_k}, \quad 0 \leq j,k \leq n, \quad x \in \mathcal{A}. $$
We already know that the set $\ddagger$ of all conjugates with the composition operation is a commutative group of order $2^n$ where each element is its own inverse \cite{Garant-Pelletier}. For the following, we enumerate elements of $\ddagger$ using the binary representation for composition of conjugates, in the same way we did for basis elements in expression (\ref{XOR algebra : binary representation}),
\begin{equation}
    \label{Conjugate hypercomplex algebra : binary representation}
    \ddagger_p := \prod_{k=0}^{n-1} \dagger_{k+1}^{p_k}, \quad \dagger_j^0 := \dagger_0 = \mathrm{id}, \quad j = 1,\ldots,n;
\end{equation}
where $\mathrm{bin}_n(p) = (p_{n-1},p_{n-2},\ldots,p_0)$. In this way, the same kind of analysis done in previous sections to extract a multiplier and index functions can be performed. 
\begin{remark}
    When $\mathcal{A}$ is a multicomplex algebra, the field is the complex numbers $\mathbb{C}$ and the usual complex conjugate is not accounted for in $\ddagger$ (since the imaginary unit $i$ is not part of the principal units set $\{1,u_1,\ldots,u_n\}$). In this case we will simply set $\ddagger_* := \ddagger \cup \dagger_* \circ \ddagger$ where $\dagger_*$ is another notation for the usual complex conjugate $z^{\dagger_*} := z^* \equiv \overline{z}, z \in \mathbb{C}$. We have defined $\dagger_* \circ \ddagger := \{ \dagger_* \circ \ddagger_p \mid \ddagger_p \in \ddagger \}$.
\end{remark}
\begin{definition}
    \label{definition : conjugates hypercomplex algebra}
    Let $x \in \mathcal{A}$ and let $+$ be the addition operation between compositions of conjugates $\ddagger_p, \ddagger_q \in \ddagger$ such that
    $$ x^{a\ddagger_p + \ddagger_q} := (x^a)^{\ddagger_p} x^{\ddagger_q}, \quad \forall a \in \mathbb{R} $$
    when the expression $x^a$ makes sense. Then we call the resulting structure $(\ddagger, +, \circ)_x$ the conjugate hypercomplex system on $x \in \mathcal{A}$. When $\mathcal{A}$ has complex coefficients, we extend this definition to $(\ddagger_*, +, \circ)$.
\end{definition}


It may seem somewhat artificial to introduce the subject in this way, but let us see how it arises in a more familiar context. For the complex numbers, we have two principal units $\{\dagger_0, \dagger_*\}$ where $\dagger_0$ is the identity and $\dagger_*$ is the usual complex conjugate. We will see below that the conjugate algebra can be identified as a multiperplex algebra, and by such a diagonal basis is directly obtained from the previous results (section \ref{subsection : Sylvester-Hadamard matrices}) :
$$ \tilde \dagger_0 = \frac{1}{2}(\dagger_0 + \dagger_*), \quad \tilde \dagger_1 = \frac{1}{2}(\dagger_0 - \dagger_*), \quad \tilde \dagger_0 \circ \tilde\dagger_1 = 0, \quad \tilde \dagger_i^2 = \tilde \dagger_i, \quad i = 1,2. $$
A property of the diagonal bases for multiperplex and multicomplex algebras is that the sum of their elements always equal to the identity \cite{courchesneAndTremblay2025}. Here we see that $\tilde \dagger_0 + \tilde \dagger_1 = \dagger_0$, and since $z^{\dagger_0} = z$ for all $z \in \mathbb{C}$, we should get
$$ z = z^{\dagger_0} = z^{\tilde\dagger_0 + \tilde \dagger_1} = z^{\tilde \dagger_0} z^{\tilde \dagger_1}, \quad \forall z \in \mathbb{C}. $$
Let's verify this claim by applying definiton \ref{definition : conjugates hypercomplex algebra}. Let $z = x + iy$, $x,y \in \mathbb{R}$, then
$$ z^{\tilde \dagger_0} = z^{\frac{1}{2}(\dagger_0 + \dagger_*)} = (z^{\frac{1}{2}})^{\dagger_0}(z^{\frac{1}{2}})^{\dagger_*} = z^{\frac{1}{2}}\overline{z}^{\frac{1}{2}} = |z| $$
and
$$ z^{\tilde\dagger_1} = z^{\frac{1}{2}(\dagger_0 - \dagger_*)} = (z^{\frac{1}{2}})^{\dagger_0}(z^{-\frac{1}{2}})^{\dagger_*} = \left(\frac{z}{\overline{z}}\right)^{\frac{1}{2}} = \frac{x + iy}{|z|} = e^{i\theta} $$
where $\theta := \arctan(y/x)$. This means that $z=z^{\tilde\dagger_0 + \tilde\dagger_1}$ is not only a true statement, but also another way to express the trigonometric representation of complex numbers.

\subsection{Conjugates on the multiperplex algebras}

For simplicity, we restrict our discussion to the case of a conjugate algebra $(\ddagger, +, \circ)$ defined over the multiperplex algebra $\mathbb{D}_n$. Since $\mathbb{D}_n$ is related to the multicomplex algebra $\mathbb{M}_n$ through a complexification, the same analysis can be extended directly to the conjugates $(\ddagger_*, +,\circ)$ of $\mathbb{M}_n$.

As a reminder, $\mathbb{D}_n$ has a set $\{1, \mathrm{j}_1, \mathrm{j}_2, \ldots, \mathrm{j}_n\}$ of principal units with $n \geq 0$ and $\mathrm{j}_k^2 = 1$. The standard basis $\{ \mathrm{e}_p \}$ has $2^n$ elements, also called the composite units, that can be written using the binary representation (\ref{XOR algebra : binary representation}):
$$ \mathrm{e}_p := \prod_{k=0}^{n-1} \mathrm{j}_{k+1}^{p_k}, \quad \mathrm{bin}_n(p) = (p_{n-1}, \ldots, p_0) \in \{ 0,1 \}^n. $$
From both theorem \ref{theorem : algorithm-like diagonal basis} and \cite{courchesneAndTremblay2025}, we know that the diagonal basis elements are of the form
$$ \Gamma_n^{\ddagger_p} := (\gamma_1 \ldots \gamma_n)^{\ddagger_p}, \quad \ddagger_p \in \ddagger, \quad \text{where} \quad \gamma_i := \frac{1}{2}(1 + \mathrm{j}_i). $$
We denote the diagonal basis $\mathcal{E}$ and its elements $\varepsilon_p := \Gamma_n^{\ddagger_p}$, giving the canonical idempotent representation for all $x \in \mathbb{D}_n$: 
$$ x = \sum_{p=0}^{2^n-1} x_{\hat p} \varepsilon_p, \quad x_{\hat p} \in \mathbb{R} $$
with the properties
$$ \varepsilon_p \varepsilon_q = \delta_{pq} \quad \text{and} \quad \sum_{p=1}^{2^n} \varepsilon_p = 1. $$
Note that a given conjugate $\dagger_k$ acts only on $\gamma_k$ (with the same index), so
\begin{equation}
    \label{equation : binary representation for diagonal}
    \varepsilon_p = (\gamma_1 \ldots \gamma_n)^{\ddagger_p} = \prod_{k=0}^{n-1} \gamma_{k+1}^{\dagger^{p_k}}
\end{equation}
where we used a simplified notation $\gamma_i^{\dagger} := \frac{1}{2}(1 - \mathrm{j}_i)$ with $\dagger^0 := \dagger_0$. This is a variation of the binary representation now appearing in the expression of diagonal basis elements, and useful to prove many properties.
\begin{proposition}
    \label{proposition : diagonal basis and conjuagate XOR}
    Let $\varepsilon_p \in \mathcal{E}$ and $\ddagger_q \in \ddagger$, then $\varepsilon_p^{\ddagger_q} = \Gamma_n^{\ddagger_{p \oplus q}} = \varepsilon_{p \oplus q}.$
\end{proposition}
\begin{proof}
    We get directly from (\ref{Conjugate hypercomplex algebra : binary representation}) and (\ref{equation : binary representation for diagonal})
    $$ \varepsilon_p^{\ddagger_q} = \left( \prod_{k=0}^{n-1} \gamma_{k+1}^{\dagger^{p_k}} \right)^{\ddagger_q} = \prod_{k=0}^{n-1} \gamma_{k+1}^{\dagger^{p_k} \circ \dagger^{q_k}} = \prod_{k=0}^{n-1} \gamma_{k+1}^{\dagger^{p_k \oplus q_k}} = \varepsilon_{p \oplus q}. $$
\end{proof}
\begin{corollary}
    \label{corollary : diagonal basis invariant}
    The basis $\mathcal{E}$ is invariant under conjugate operations and for all $r=0,\ldots,2^{n}-1$, $\ddagger_r:\mathcal{E} \rightarrow \mathcal{E}$ is a bijection.
\end{corollary}
\begin{proof}
    From proposition \ref{proposition : diagonal basis and conjuagate XOR}, $\mathcal{E}$ is clearly invariant under any conjugate operation. Now let $\ddagger_r \in \ddagger$, and take $\varepsilon_q \in \mathcal{E}$, then
    $$ \varepsilon_{q \oplus r}^{\ddagger_r} = \varepsilon_{q \oplus r \oplus r} = \varepsilon_{q} $$
    so $\ddagger_r$ is surjective. To prove injectivity, let $\varepsilon_p, \varepsilon_q \in \mathcal{E}$ such that $\varepsilon_p^{\ddagger_r} = \varepsilon_q^{\ddagger_r}$, then
    $$ \varepsilon_{p \oplus r} = \varepsilon_{q \oplus r} \Rightarrow p \oplus r = q \oplus r \Rightarrow p = q. $$
\end{proof}
\begin{definition}
    \label{definition : exponentiation in multi-algebra}
    We define the exponentiation for $x \in \mathbb{D}^n$ and $a \in \mathbb{R}$ as
    $$ x^a := \sum_{k=0}^{2^n-1} x_{\hat k}^a \varepsilon_k, \quad x_{\hat{k}} \in \mathbb{R}, $$
    whenever this expression makes sense i.e. $x_{\hat{k}} \neq 0$ for all $k$ if $a < 0$.
\end{definition}
\begin{proposition}
    \label{proposition : basic properties conjugates}
    For all $x \in \mathbb{D}_n$, $\ddagger_p, \ddagger_q \in \ddagger$ and $a \in \mathbb{R}$, we have
    \begin{enumerate}
        \item $x^{\ddagger_p \circ \ddagger_q} = x^{\ddagger_q \circ \ddagger_p}$,
        \item $(x^{a})^{\ddagger_p} = (x^{\ddagger_p})^a$,
        \item if $x$ is invertible, then $x^{\ddagger_p}$ is also invertible.
    \end{enumerate}
\end{proposition}
\begin{proof}
    First,
    $$ \begin{aligned}
        x^{\ddagger_p \circ \ddagger_q} &= (x^{\ddagger_p})^{\ddagger_q} = \left( \sum_{k=0}^{2^n-1} x_{\hat k} \varepsilon_k^{\ddagger_p} \right)^{\ddagger_q}\\
        &= \left( \sum_{k=0}^{2^n-1} x_{\hat k} \varepsilon_{k \oplus p} \right)^{\ddagger_q} = \sum_{k=0}^{2^n-1} x_{\hat k} \varepsilon_{k \oplus p \oplus q}\\
        &= \sum_{k=0}^{2^n-1} x_{\hat k} \varepsilon_{k \oplus q \oplus p} = \left( \sum_{k=0}^{2^n-1} x_{\hat k} \varepsilon_k^{\ddagger_q} \right)^{\ddagger_p}\\
        &= (x^{\ddagger_q})^{\ddagger_p} = x^{\ddagger_q \circ \ddagger_p}.
    \end{aligned} $$
    For the second property,
    $$ (x^{\ddagger_p})^a = \left(\sum_{k=0}^{2^n-1} x_{\hat k} \varepsilon_{k}^{\ddagger_p} \right)^{a} $$
    and since from corollary \ref{corollary : diagonal basis invariant}, $\ddagger_p : \mathcal{E} \rightarrow \mathcal{E}$ is a bijection, definition \ref{definition : exponentiation in multi-algebra} still applies in the same way:
    $$ (x^{\ddagger_p})^a = \sum_{k=0}^{2^n-1} x_{\hat k}^a \varepsilon_{k}^{\ddagger_p} = (x^a)^{\ddagger_p}. $$
    For the third property, we note that $x$ is invertible if and only if its coefficients in the diagonal basis are nonzero \cite{courchesneAndTremblay2025}. Let $x_{\hat k} \neq 0$ for all $k$, then
    $$ x^{\ddagger_p} = \sum_{k=0}^{2^n-1} x_{\hat k} \varepsilon_{k \oplus p}. $$
    Once again, since $\ddagger_p$ is a bijection from $\mathcal{E}$ onto itself, the set $\{ \varepsilon_{k \oplus p} \}_k$ is just another way to enumerate the basis elements, and thus the resulting element $x^{\ddagger_p}$ is also invertible.
\end{proof}
\begin{remark}
    This last proof relies on the fact that $\ddagger_p$ is distributive over the addition and multiplication operations in $\mathbb{D}_n$, which is due directly from the fact that conjugation operations act only on principal units $\mathrm{j}_1, \ldots, \mathrm{j}_n$, no matter where they are.
\end{remark}
\begin{theorem}
    The structure $(\ddagger, +, \circ)_x$, as defined in definition \ref{definition : conjugates hypercomplex algebra} on the \textbf{invertible elements} $x \in \mathbb{D}_n$, satisfies all the properties of an algebra i.e. $(\ddagger, +)$ is a vector space and $\circ$ is bilinear.
\end{theorem}
\begin{proof}
    First, $(\ddagger,+)$ must satisfy the vector space axioms. Let $x \in \mathbb{D}_n$ be an invertible element, $\ddagger_p, \ddagger_q, \ddagger_r \in \ddagger$ and $a,b \in \mathbb{R}$. 

    Associativity of the addition:
    $$ x^{\ddagger_p + (\ddagger_q + \ddagger_r)} = x^{\ddagger_p} x^{\ddagger_q + \ddagger_r} = x^{\ddagger_p} (x^{\ddagger_q} x^{\ddagger_r}) = (x^{\ddagger_p} x^{\ddagger_q}) x^{\ddagger_r} = x^{(\ddagger_p + \ddagger_q) + \ddagger_r}. $$

    Commutativity of the addition:
    $$ x^{\ddagger_p + \ddagger_q} = x^{\ddagger_p} x^{\ddagger_q} = x^{\ddagger_q} x^{\ddagger_p} = x^{\ddagger_q + \ddagger_p}. $$

    Identity element for the addition. Here $a \equiv a\dagger_0$ in $(\ddagger, +, \circ)_x$:
    $$ x^{\ddagger_p + 0} = x^{\ddagger_p} x^0 = x^{\ddagger_p}. $$

    Inverse element for the addition (using property 3 of proposition \ref{proposition : basic properties conjugates}):
    $$ x^{\ddagger_p - \ddagger_p} = x^{\ddagger_p} x^{-\ddagger_p} = \frac{x^{\ddagger_p}}{x^{\ddagger_p}} = 1 = x^0. $$

    Compatibility of scalar multiplication with field multiplication:
    $$ x^{a(b\ddagger_p)} = (x^a)^{b\ddagger_p} = ((x^{a})^b)^{\ddagger_p} = (x^{(ab)})^{\ddagger_p} = x^{(ab)\ddagger_p}. $$

    Identity element for the scalar multiplication:
    $$ x^{1\ddagger_p} = (x^1)^{\ddagger_p} = x^{\ddagger_p}. $$

    Distributivity of scalar multiplication over the vector addition:
    $$ x^{a(\ddagger_p + \ddagger_q)} = (x^a)^{\ddagger_p + \ddagger_q} = (x^{a})^{\ddagger_p}(x^a)^{\ddagger_q} = x^{a\ddagger_p + a\ddagger_q}. $$

    Distributivity of scalar multiplication over the field addition:
    $$ x^{(a+b)\ddagger_p} = (x^{a+b})^{\ddagger_p} = (x^ax^b)^{\ddagger_p} = x^{a\ddagger_p} x^{b\ddagger_p} = x^{a\ddagger_p + b\ddagger_p}. $$
    Then, we prove that $\circ$ is bilinear.

    Right distributivity of $\circ$ (left distributivity can be shown in a similar way):
    $$ x^{(a \ddagger_p + b\ddagger_q)\circ\ddagger_r} = (x^{a\ddagger_p} x^{b\ddagger_q})^{\ddagger_r} = x^{a\ddagger_p\circ\ddagger_r}x^{b\ddagger_q\circ\ddagger_r} = x^{a\ddagger_p\circ\ddagger_r + b\ddagger_q\circ\ddagger_r}. $$

    Compatibility of $\circ$ with scalars (using property 2 of proposition \ref{proposition : basic properties conjugates}):
    $$ x^{(a\ddagger_p)\circ(b\ddagger_q)} = (x^{a\ddagger_p})^{b\ddagger_q} = (((x^a)^{\ddagger_p})^{b})^{\ddagger_q} = (((x^a)^{b})^{\ddagger_p})^{\ddagger_q} = x^{(ab)\ddagger_p \circ \ddagger_q}. $$
    This completes the proof.
\end{proof}
\begin{remark}
    The proof has been done only on basis elements $\ddagger_0,\ddagger_1,\ldots,\ddagger_{2^n-1}$ of $(\ddagger, +, \circ)$, an arbitrary element would be of the form:
    $$ \eta = \sum_{p=0}^{2^n-1} \eta_p \ddagger_p, \quad \eta_p \in \mathbb{R}. $$
    Since all the properties of an algebra are satisfied for the basis elements, it extends directly to arbitrary elements of $\eta \in (\ddagger, +, \circ)_x$.
\end{remark}
\begin{corollary}
    Let $x \in \mathbb{D}_n$ be an invertible element. Then $(\ddagger, +, \circ)_x$ has the structure of a multiperplex algebra, has a diagonal basis $\{ \tilde \ddagger_p \}_{p=0}^{2^n-1}$, and $x$ has the following representation:
    $$ x = x^{\sum_p \tilde \ddagger_p} = \prod_p x^{\tilde \ddagger_p}. $$
\end{corollary}

\bibliographystyle{plain}
\bibliography{refs}

\vspace*{1cm}

\noindent $^1$D\'EPARTEMENT DE MATHÉMATIQUES ET D'INFORMATIQUE, \\ UNIVERSITÉ DU QUÉBEC, TROIS-RIVI\`ERES, QC, CANADA \\
{\em Email address}: \texttt{derek.courchesne@uqtr.ca} \\
\\
\noindent $^2$D\'EPARTEMENT DE MATHÉMATIQUES ET D'INFORMATIQUE, \\ UNIVERSITÉ DU QUÉBEC, TROIS-RIVI\`ERES, QC, CANADA \\
{\em Email address}: \texttt{sebastien.tremblay@uqtr.ca}\\

\end{document}